\documentclass[acmsmall]{acmart}

\setcopyright{none}

\usepackage{enumitem}
\setlist{itemsep=0.5ex,parsep=0pt} 

\usepackage{fancyhdr}
\usepackage{xspace}
\usepackage{verbatim}
\usepackage{url}
\usepackage{multirow}
\usepackage{latexsym,amsthm,amsmath,amsfonts,stmaryrd,mathtools}
\usepackage{marvosym}
\usepackage{adjustbox}
\usepackage{afterpage}
\usepackage{graphicx}

\usepackage{listings}
\usepackage{wrapfig}
\usepackage{subcaption}

\usepackage{booktabs}
\usepackage{multirow}

\usepackage{natbib}
\usepackage{url}
\usepackage{tikz}
\usetikzlibrary{arrows}
\usetikzlibrary{calc}

\usepackage{wasysym}

\newcommand{\out}[1] {}

\newcommand{\cd}[1]{\lstinline`{#1}`}

\reversemarginpar
\newif\ifnotes
\notestrue

\newcommand{\punt}[1]{}

\newcommand{\secref}[1]{Section~\ref{sec:#1}}

\newcommand{\figref}[1]{Figure~\ref{fig:#1}}

\newcommand{\lineref}[1]{line~\ref{line:#1}}

\newcommand{\thmref}[1]{Theorem~\ref{thm:#1}}

\newcommand{\lemref}[1]{Lemma~\ref{lem:#1}}

\renewcommand{\eqref}[1]{Equation~(\ref{eq:#1})}

\newcounter{remark}[section]

\newcommand{\defn}[1]{\emph{\textbf{#1}}}

\definecolor{darkblue}{HTML}{0007C9}

\DeclarePairedDelimiterX{\dblangle}[1]{\langle}{\rangle}{\hspace{-0.5ex}\delimsize\langle{#1}\delimsize\rangle\hspace{-0.5ex}}

\theoremstyle{plain}

\newtheorem{corollary}{Corollary}

\theoremstyle{definition}
\newtheorem{definition}{Definition}

\newdimen\zzlistingsize
\newdimen\zzlistingsizedefault
\newdimen\kwlistingsize
\zzlistingsize=\zzlistingsizedefault
\gdef\lco{black}

\newcommand{\Lstbasicstyle}{\fontsize{\zzlistingsize}{1.1\zzlistingsize}\ttfamily\color{\lco}}
\newcommand{\keywordstyle}{\fontsize{1.09\kwlistingsize}{\kwlistingsize}\normalfont\bf\color{\lco}}
\newlength{\zzlstwidth}
\newcommand{\lcm}{\color{\lco}}

\newcommand{\code}[1]{\lstinline!#1!}

\newcommand{\pathL}{\textsf{\textbf L}}
\newcommand{\pathR}{\textsf{\textbf R}}
\newcommand{\dagdepth}[1]{\ensuremath{\mathop{\text{D}}( {#1} )}}
\newcommand{\forkpath}[1]{\ensuremath{\mathop{\text{P}}( {#1} )}}
\newcommand{\precedes}{\ensuremath{\preccurlyeq}}
\newcommand{\lcp}[2]{\ensuremath{\mathop{\text{LCP}}( {#1}, {#2} )}}

\newcommand{\concat}{%
  \mathbin{{+}\mspace{-8mu}{+}}%
}

\usepackage{math-cmds}
\usepackage{ifthen}
\usepackage{mathpartir}
\usepackage{cleveref}

\renewcommand{\infer}[3][]{
\ifthenelse{\equal{#1}{}}{
\inferrule{#2}{#3}
}
{
\inferrule*[right={\scriptsize \textbf{#1}}]
{#2}
{#3}
}
}

\crefname{lem}{lemma}{lemmas}
\Crefname{lem}{Lemma}{Lemmas}
\crefname{thm}{theorem}{theorems}
\Crefname{thm}{Theorem}{Theorems}
\crefname{figure}{figure}{figures}
\Crefname{figure}{Figure}{Figures}

\usepackage{enumitem}
\setlist{itemsep=0.05ex,parsep=0pt}             

\begin{document}

\title{DePa: Simple, Provably Efficient, and Practical Order Maintenance for
Task Parallelism}


\author{Sam Westrick}
\affiliation{
  \department{Department of Computer Science}              
  \institution{Carnegie Mellon University}            
  \city{Pittsburgh}
  \state{PA}
  \country{USA}                    
}
\email{swestric@cs.cmu.edu}          

\author{Larry Wang}
\affiliation{
  \department{Department of Computer Science}              
  \institution{Carnegie Mellon University}            
  \city{Pittsburgh}
  \state{PA}
  \country{USA}                    
}
\email{lawrenc2@alumni.cmu.edu}

\author{Umut A. Acar}
\affiliation{
  \department{Department of Computer Science}              
  \institution{Carnegie Mellon University}            
  \city{Pittsburgh}
  \state{PA}
  \country{USA}                    
}

\begin{abstract}
A number of problems in parallel computing require reasoning about the dependency structure in parallel programs. 
For example, dynamic race detection relies on efficient ``on-the-fly''
determination of dependencies between sequential and parallel tasks
(e.g. to quickly determine whether or not two memory accesses occur in
parallel).
Several solutions to this ``parallel order maintenance'' problem has
been proposed, but they all have several drawbacks, including lack of
provable bounds, high asymptotic or practical overheads, and poor support for parallel execution.

In this paper, we present a solution to the parallel order maintenance problem that is provably efficient, fully parallel, and practical.
Our algorithm---called \defn{DePa}---represents a computation as a graph and
encodes vertices in the graph with two components: a \emph{dag-depth}
and a \emph{fork-path}.
In this encoding, each query requires $O(f/\omega)$ work, where
$f$ is the minimum dynamic nesting depth of the two vertices compared, and
$\omega$ is the word-size.
%
%
%
In the common case (where $f$ is small, e.g., less than 100), each
query requires only a single memory lookup and a small constant number of
bitwise instructions.
Furthermore, graph maintenance at forks and joins requires only constant work,
resulting in no asymptotic impact on overall work and span.
DePa is therefore work-efficient and fully parallel.

\end{abstract}




\maketitle

\section{Introduction}
\label{sec:intro}
With the mainstream availability of multicore computers, parallel
programming today is important, relevant, and challenging.
One of the chief factors that make it challenging is race conditions,
which are difficult to avoid, debug, and fix.
To help address this challenge, researchers have developed a variety of
\emph{dynamic race detection} algorithms which instrument a program to
identify data races dynamically, i.e., while the program is
running~\cite{SavageBuNe97,FlanaganFr09,cc-hybrid-2003,yrc-2005-racetrack,Mellor-Crummey91,FengLe97,ChengFeLe98,BenderFiGi04,Fineman05,RamanZhSa10,RamanZhSa12,uaf+race-2016,xsl-determinacy-2020}.
At a high level, these approaches broadly fall into two camps: those targeting
general, coarse-grained
concurrency~\cite{SavageBuNe97,FlanaganFr09,cc-hybrid-2003,yrc-2005-racetrack},
and those targeting fine-grained, structured parallelism~\cite{Mellor-Crummey91,FengLe97,ChengFeLe98,BenderFiGi04,Fineman05,RamanZhSa10,RamanZhSa12,uaf+race-2016,xsl-determinacy-2020}.
In this paper, we focus on the latter, specifically for nested fork-join
parallelism (e.g. nested parallel-for loops), provided by
systems such as
Cilk~\cite{frigolera98,blumofe96,muller+responsive-2020},
ParlayLib~\cite{bad+parlaylib-2020},
OpenMP~\cite{openmp18},
Microsoft TPL~\cite{tpl09},
Intel TBB~\cite{threadingbuildingblocksmanual},
Fork/Join Java~\cite{lea00},
various forms of
Parallel ML~\cite{manticore-implicit-11,OTU-smlsharp-2018,rmab-mm-2016,mah-responsive-2017,gwraf-hieararchical-2018,westrick+disent-2020,awa+space-2021},
and many others.
Programs written in these languages and libraries
typically consist of many (e.g., millions of)
fine-grained tasks which are spawned and synchronized in a structured fashion.

At the heart of existing race-detection algorithms is an
\defn{order maintenance} data structure which
encodes sequential dependencies between
tasks~\cite{Mellor-Crummey91,FengLe97,ChengFeLe98,BenderFiGi04,Fineman05,RamanZhSa10,RamanZhSa12,uaf+race-2016,xsl-determinacy-2020,xal-determinacy-2021}.
The purpose of this data structure is to enable efficient queries, to
quickly determine
whether or not two tasks are logically concurrent with one another.%
\footnote{Two tasks are \emph{logically concurrent} if there exists
a schedule in which their operations could be interleaved.}
Race detection then can be performed by monitoring individual reads and writes:
a data race is detected when two (logically) concurrent tasks both access the
same memory location, and at least one of these tasks modifies the location.

In practice, the overheads of race detection are largely determined by
the performance characteristics (and implementation details) of this
order maintenance structure.
As such, many solutions have been developed, specifically with the goals of
low overhead both in theory and practice, and support for parallel execution.
However, existing solutions come with a range of drawbacks.
Some do not support parallel execution, or only provide limited support for
it~\cite{FengLe97,RamanZhSa10,BenderFiGi04}.
Others support parallel execution but
incur either high theoretical overhead per query~\cite{Mellor-Crummey91} or
present no bound on this cost~\cite{RamanZhSa12}.
For nested fork-join programs, Utterback et al.'s algorithm~\cite{uaf+race-2016}
is asymptotically optimal; however, it is complex and difficult to implement,
as it requires a custom scheduler which is tailored specifically to order
maintenance.
We therefore seek a solution which has provably low overhead, is fully parallel,
and is simple to implement in practice.

Our solution is a ``depths-and-paths'' data structure,
or \defn{DePa} for short, which encodes the computation as a
graph and performs precendence (reachability) queries between vertices.
The DePa structure maintains only the dynamic task tree instead
of the whole computation graph, and annotates each tree node with two
quantities: (1) its \defn{depth} in the graph, and (2) its \defn{path} from
the root of the task tree.
We then prove that the depth and the path are sufficient
for answering precedence queries on the computation graph.
Essentially, the depth and path serve as two ``coordinates'' in a computation
graph: together, these not only uniquely identifying a vertex, but
also enable us to efficiently determine whether or not there is a path
between two vertices.

For most computations the depth and path can fit into a few machine
words and therefore the cost of precedence queries is a small constant.
In particular, we prove
(\thmref{query-cost})
that our algorithm requires only $O(\min(f_u,f_v)/\omega)$ work for a query
between two graph vertices $u$ and $v$.
Here, $f_u$ and $f_v$ are (respectively) the dynamic nesting depths
of $u$ and $v$, and $\omega$ is the word size of the machine.
This is effectively a constant bound in practice: with proper
load-balancing and granularity control, dynamic nesting depths are typically
small (e.g., less than 100).
In terms of space, we show that the overhead of the algorithm is small:
storing the depth and path of a vertex $v$ requires only $O(f_v/\omega)$ space,
and any additional data required by the algorithm is bounded by
$O(PF + PF/\omega)$ where $P$ is the number of processors and
$F$ is the maximum dynamic nesting depth of the computation.

In addition to supporting queries efficiently, our
DePa algorithm is fully parallel.
Maintaining the task tree requires only local computation with
constant cost at each fork and join to update depths and paths.
The work and span of a program therefore stays asymptotically unchanged by
DePa maintenance, and hence so does its parallelism (the ratio of work to
span).
Queries are also parallel, as they can be implemented almost entirely
with simple bitwise operations and arithmetic, and do not require
any synchronization.
Finally, to improve efficiency in practice, we present an optimization
technique which integrates DePa with a work-stealing scheduler to
improve data locality at queries.
This optimization---although not essential---is easy to implement and requires
only small changes to the scheduler itself.

In summary, the contributions of this paper include the following.
\begin{itemize}
\item The ``depth-and-paths'' (DePa) data structure for representing and querying task graphs for series-parallel ordering, along with a correctness proof (\secref{sp-order}).

\item Asymptotic analysis establishing work, span, and space bounds
for DePa (\secref{cost-analysis}).

\item Optimization techniques for improving practical efficiency by
    integrating with work-stealing scheduling (\secref{sched}).
\end{itemize}

\section{Preliminaries}
\label{sec:prelim}

\paragraph{Fork-join}
We consider \defn{fork-join} parallel programs which are based on
\defn{tasks} organized in a dynamic \defn{task tree}.
Initially, there is a single ``root'' task.
At any moment, any leaf (a task with no children) may \defn{fork}, which
creates two child tasks.
While the children are executing, the parent task is suspended.
As soon as both children have completed, they \defn{join} with the parent,
which deletes the children from the tree and resumes execution of the parent
task.
At any point of program execution, the tasks at leaf nodes of the task tree
are said to be \defn{running} while the tasks at internal nodes are said to be
\defn{suspended}.
Any task which is either suspended or running is said to be \defn{active}.
We say that two tasks are \defn{concurrent} if neither is an ancestor of the
other.
That is, in the task tree, concurrent tasks could be siblings, cousins, etc.

\paragraph{Dags (Computation Graphs)}
A fork-join computation can be summarized with a directed, acylic graph called
a \defn{dag}~\cite{acrs-dag-calculus-2016}.
The dag represents the history of a completed computation:
each vertex represents a sequence of executed instructions uninterrupted by
a fork or join,
and edges are ordering constraints.
We can construct a dag dynamically during execution by (1) creating new outgoing
edges to two new vertices at each fork, and (2) creating two incoming edges to a
single new vertex at each join.
We say that vertex $u$ \defn{precedes} vertex $v$, denoted
$u \preccurlyeq v$, if there exists a path in the dag from $u$ to $v$.
Note that $\preccurlyeq$ is a partial order.
%

\paragraph{Series-Parallel Dags}
The computation graphs of (nested) fork-join parallel programs are
\emph{series-parallel} graphs, i.e. inductively
constructed by sequential and parallel composition.
This is illustrated in
\figref{series-parallel}.

\section{DePa Algorithm}
\label{sec:sp-order}
We present an order maintenance data structure which efficiently
supports queries of the form $u \preccurlyeq_G v$, i.e., determining whether or
not there is a path from $u$ to $v$ in computation graph $G$.
At a high level, our algorithm assigns each vertex a unique
\defn{vertex identifier} consisting of two components:
a \emph{dag-depth} and a \emph{fork-path}.
Together, the dag-depths and fork-paths of vertices allow us to identify
vertices and efficiently determine their relative positions in the dag.
We call our algorithm \textbf{DePa} for its combined use of \textbf{de}pths
and \textbf{pa}ths.

\subsection{Dag-Depths and Fork-Paths}
The \defn{dag-depth} of a vertex
is the length of the longest path that ends at that vertex.
%
The \defn{fork-path} of a vertex encodes the
nesting of its corresponding task as a sequence of bits.
Each bit is one of \pathL{} or \pathR{}, to respectively indicate a fork either
on the left or the right.
We denote dag-depths with $\dagdepth v$ and fork-paths with $\forkpath v$.
Furthermore, we denote the longest common prefix between two fork-paths
with $\lcp{p_1}{p_2}$.

Dag-depths and fork-paths can be computed solely in terms of
a vertex's immediate predecessors, as shown in \figref{strand-id-fork-join}.
When a vertex $u$ forks into $v$ and $w$, the new vertices have depth
$\dagdepth {v} = \dagdepth {w} = 1 + \dagdepth u$, and their
paths are extended with $\pathL$ and $\pathR$; that is,
$\forkpath {v} = \forkpath u \concat \pathL$ and
$\forkpath {w} = \forkpath u \concat \pathR$
(where $\concat$ denotes concatenation).
When two vertices $u$ and $v$ join, they must be siblings, i.e. their
paths must be $p \concat \pathL$ and $p \concat \pathR$ for some common prefix $p$.
The new vertex then has path $p$ and depth
$1 + \max( \dagdepth {u}, \dagdepth {v} )$.
An example of a full dag is shown in \figref{dag-paths-depths}, where
we write ``---'' for the empty path.

\begin{figure}
\begin{minipage}{0.35\textwidth}
\input{fig/fig-series-parallel}
\end{minipage}
\hfill
\begin{minipage}{0.32\textwidth}
\input{fig/fig-strand-id-fork-join}
\end{minipage}
\hfill
\begin{minipage}{0.25\textwidth}
\input{fig/fig-dag-paths-depths}
\end{minipage}
\end{figure}

Intuitively, it's helpful to think of the dag-depth and fork-path as two
``coordinates'' into a dag.
When illustrating a dag with edges pointing down, the dag-depth serves as a
vertical position, and the fork-path is a horizontal position.
Together, these uniquely identify a vertex.%
\footnote{We leave this to intuition, and
prove a more general result in \lemref{precedes-lca}.}

\subsection{Identifying Precedence}
When two vertices have the same fork-path, one must precede the other
in the dag.
More generally, whenever one vertex's fork-path is a prefix of another, one of
the two vertices must precede the other.
This precendence can be determined by comparing their dag-depths,
as stated in \lemref{precedes-lca}.

\begin{lemma}
\label{lem:precedes-lca}
Let $u$ and $v$ be any two vertices in dag $G$ such that
$\forkpath u$ is a prefix of $\forkpath v$.
Then (1) $u \precedes_G v$ if and only if $\dagdepth u \le \dagdepth v$,
and (2) $v \precedes_G u$ if and only if $\dagdepth v \le \dagdepth u$.
\end{lemma}
\begin{proof}
We will only present a proof of (1), as (2) is symmetric.
The forward direction is straightforward: if $u \precedes_G v$, then there
exists a path in the dag from $u$ to $v$, so $\dagdepth u \le \dagdepth v$.
To prove the reverse direction, we induct over the structure of $G$.
Consider $u$ and $v$ such that
$\forkpath u$ is a prefix of $\forkpath v$ and
$\dagdepth u \le \dagdepth v$.
There are three cases for the structure of $G$:
\begin{itemize}
\item
If $G$ is a single vertex, the lemma holds trivially.

\item
Suppose $G$ is the serial composition of $G_1$
and $G_2$.
If $u$ and $v$ either both occur in $G_1$ or both occur in
$G_2$, then $u \precedes_G v$ by induction.
Otherwise, due to $\dagdepth u \leq \dagdepth v$, we must have that $u$ occurs
in $G_1$ and $v$ occurs in $G_2$.
Consider the vertex $r$ where the two sub-dags overlap.
This vertex $r$ is the sink of $G_1$, so $u \precedes_{G_1} r$.
Similarly, $r$ is the source of $G_2$, so $r \precedes_{G_2} v$.
Together, these imply $u \precedes_G v$.

\item
Suppose $G$ is the parallel composition of $G_1$
and $G_2$.
If $u$ is the source vertex of $G$ or if $v$ is the sink vertex
of $G$, then clearly we have $u \precedes_G v$.
If $u$ and $v$ either both occur in $G_1$ or both occur in
$G_2$, then $u \precedes_G v$ by induction.
Otherwise, suppose WLOG that $u$ occurs in $G_1$ and
$v$ occurs in $G_2$.
This violates the assumption that $\forkpath u$ is a prefix of $\forkpath v$,
because all fork-paths of vertices in $G_1$ begin with
$p \concat \pathR$ and all in $G_2$ begin with
$p \concat \pathL$ (where $p$ is the fork-path of the source of $G$).
\end{itemize}
\end{proof}

A direct application of this lemma reveals that our vertex identifiers do indeed
uniquely identify vertices.

\begin{corollary}
For any two vertices $u$ and $v$, $u = v$ if and only if
$\forkpath u = \forkpath v$ and $\dagdepth u = \dagdepth v$.
\end{corollary}

We can also obtain a method for checking precedence
between any two arbitrary vertices.
To do so, we identify a special \defn{critical} vertex, denoted
$\textsf{C}(u, v)$, defined below.
Note that the critical vertex is always well-defined for any pair of vertices,
because for every prefix of a vertex's fork-path there exists a vertex
with less-or-equal dag-depth.

\begin{definition}[Critical Vertex]
Consider any two vertices $u$ and $v$, and WLOG assume
$\dagdepth u \leq \dagdepth v$.
The critical vertex $\textsf{C}(u,v)$
is defined as the vertex $c$ such that
$\forkpath c = \lcp {\forkpath u} {\forkpath v}$
and $\dagdepth c$ is maximized s.t. $\dagdepth c \leq \dagdepth v$.
\end{definition}

As formalized below in \lemref{precedes}, the critical vertex is special
in that if a path $u \precedes v$ exists, then
$\textsf{C}(u,v)$ will certainly lie on the path.
Additionally, if there is not a path between the two vertices, then the
critical vertex serves as proof that the path does not exist.

\begin{lemma}
\label{lem:precedes}
Let $u$ and $v$ be any two vertices, and let $c = \textsf{C}(u, v)$ be their
critical vertex. Then $u \precedes v$
if and only if $\dagdepth u \leq \dagdepth c \leq \dagdepth v$.
\end{lemma}
\begin{proof}
We separately consider the ``if'' and ``only-if'' directions.
For the ``if'' direction, suppose
$\dagdepth u \leq \dagdepth c \leq \dagdepth v$.
Then by \lemref{precedes-lca}, we have both $u \precedes c$ and $c \precedes v$,
and therefore $u \precedes v$.
Note that the Lemma is applicable because $\forkpath c$ is a prefix of
both $\forkpath u$ and $\forkpath v$.
For the ``only-if'' direction, suppose $u \precedes v$, i.e. that there exists
a path from $u$ to $v$ in the dag.
Consider the sequence of vertices in order along the path from $u$ to $v$.
For each adjacent pair of vertices, their fork-paths differ only by
either pushing or popping one element at the end.
Therefore, at least one vertex along the path has fork-path
$\lcp {\forkpath u} {\forkpath v}$.
The last such vertex along the path is $c$, as it has maximum dag-depth.
Because $c$ is on the path, and because dag-depths strictly increase along
each path, we have $\dagdepth u \leq \dagdepth c \leq \dagdepth v$.
\end{proof}

\paragraph{Relationship with Task Trees}
Observe that whenever a leaf task suspends (i.e. when it forks),
at that moment its current vertex has larger dag-depth than any task which has
ever had the same fork-path.
The vertices of suspended ancestors are therefore critical vertices for their
children in the task tree, as these vertices cover all possible prefixes
of the fork-path and have maximum dag-depths for those fork-paths.
In particular, for queries of the form $u \precedes_G v$ where $v$ is the
current vertex of a leaf task (which are exactly the sort of queries
needed for entanglement detection), the critical vertex $\textsf{C}(u,v)$ will
be the current vertex of one of the task's suspended ancestors.

\subsection{Algorithm}

The DePa order maintenance algorithm is presented in
\figref{algo-order-maintenance}.
It consists of three components: forks, joins, and precedence queries.
As described in \lemref{precedes}, precedence queries rely on knowing the
depths of critical vertices.
Accordingly, the algorithm maintains a global lookup data structure called
$\texttt{SuspendedDepths}$ which maps fork-paths of suspended ancestors to
their dag-depths.
(In theory, this could be implemented simply as a hash table; we provide a
detailed discussion of an optimized implementation in practice to
\secref{sched}.)

At each fork, the algorithm initializes new vertex identifiers for the two new
tasks, and updates the $\texttt{SuspendedDepths}$ data structure to
remember the dag-depth of the task that forked (which is now suspended).
At each join, in preparation of resuming the parent task,
the algorithm constructs a new vertex identifier for the parent and
removes the old value from the $\texttt{SuspendedDepths}$ data structure.
At each query, the algorithm computes the longest-common-prefix of the
vertices' fork paths and looks up the appropriate dag-depth to infer
precedence.

\begin{figure}
\input{fig/fig-algo-order-maintenance}
\end{figure}

\begin{theorem}[DePa Correctness]
Let $v$ be the current vertex of some leaf task.
Then for any other vertex $u$, the order maintenance algorithm will
correctly report that $u$ precedes $v$ if and only if $u \precedes v$.
\end{theorem}

\begin{proof}
Let $p = \lcp{\forkpath{u}}{\forkpath{v}}$.
There are two cases to consider: when $p = \forkpath v$ and when
$p \neq \forkpath v$.
If $p = \forkpath v$, then $\textsf{C}(u,v) = v$.
In this case, the algorithm reports (\figref{algo-order-maintenance}, \lineref{report-path-equal-to-lcp})
that there is a path according to $\dagdepth u \leq \dagdepth v$.
By \lemref{precedes}, since $\textsf{C}(u,v) = v$, this is correct.
Next, suppose $p \neq \forkpath v$.
Then there is a suspended ancestor whose fork-path is $p$; let $c$ be the
vertex of this ancestor, and note that \code{SuspendedDepths[$p$]} stores
$\dagdepth c$.
We know that $c = \textsf{C}(u,v)$ and also that
$\dagdepth c \leq \dagdepth v$ because $v$ occurs in a descendant task.
Therefore, by \lemref{precedes}, the algorithm correctly reports
(\figref{algo-order-maintenance}, \lineref{report-path-neq-to-lcp}) that there is a
path according to
$\dagdepth u \leq \texttt{SuspendedDepths[p]}$.
\end{proof}

\subsection{Cost Analysis}
\label{sec:cost-analysis}

\begin{lemma}[Vertex Identifier Space]
\label{lem:vertex-space}
Each vertex identifier can be compactly represented using $O(f/\omega)$ words,
where $f$ is the dynamic nesting depth of the corresponding task, and $\omega$
is the word-size.
\end{lemma}
\begin{proof}
The dag-depth component requires one word, and the fork-path is a bit
sequence of length $f$. Together, these require $1 + \lceil f/\omega \rceil$
words.
\end{proof}

\begin{theorem}[DePa Space]
\label{thm:space-bound}
At any point in a $P$-processor execution, the algorithm requires
$O(PF + PF/\omega)$ additional space, where $F$ is the maximum dynamic
nesting depth of the computation.
\end{theorem}
\begin{proof}
There are two sources of space overhead: (1) each running task stores its
current vertex identifier, and (2) the $\texttt{SuspendedDepths}$ data
structure.
We assume here that the scheduler guarantees at most $O(P)$ running tasks
at any moment (typical schedulers such as work-stealing have this guarantee).
Each of these stores a vertex identifier, requiring $O(P F/\omega)$ space.
In the $\texttt{SuspendedDepths}$ data structure, each running task has
at most $F$ suspended ancestors, therefore the size of
$\texttt{SuspendedDepths}$ at any moment is at most $O(PF)$.
\end{proof}

\begin{theorem}[DePa Query Cost]
\label{thm:query-cost}
A query between vertices $u$ and $v$ costs
$O(\min(f_u,f_v)/\omega)$ work, where $f_u$ and $f_v$ are the
dynamic nesting depths of $u$ and $v$, respectively.
\end{theorem}
\begin{proof}
Using the representation of \lemref{vertex-space}, computing the longest common
prefix between fork-paths requires $O(\min(f_u,f_v)/\omega)$ work.
We can perform the \texttt{SuspendedDepths} lookup within the same bound
by representing it as a hash table. The rest of the operations require
constant work.
\end{proof}

\begin{theorem}[DePa Fork and Join Cost]
\label{thm:work-and-span}
DePa requires $O(1)$ work and span at each fork and join to update depths and
paths.
\end{theorem}
\begin{proof}
By representing paths as a linked list of words, we can update paths at
forks and joins with $O(1)$ work. The \texttt{SuspendedDepths} data structure
can be implemented as a hash table with paths as keys. Each vertex can
additionally store its own hash, so that insertions and deletions (at forks
and joins) are $O(1)$.
\end{proof}

\section{Integration with Work-Stealing}
\label{sec:sched}
For SP-order maintenance to be fast in practice, it is important not
only to optimize individual checks, but also to make sure that forks
and joins are fast.
Much of the work of our order-maintenance algorithm
at forks and joins is ``local'' in the sense that it only requires accessing
the data of a task and its children, which should be fast in practice
and parallelize well.
There is one operation, however, which could potentially incur a significant
overhead: updating the \texttt{SuspendedDepths} lookup data
structure of the DePa algorithm (\figref{algo-order-maintenance}).

In the cost theorem (\thmref{work-and-span}), we assume that the
\texttt{SuspendedDepths} data-structure provides $O(1)$ updates and
deletions, with no contention.
In practice, a straightforward implementation such as a global shared hash
table may incur significant contention on concurrent accesses.
In this section, we describe how to improve efficiency by integrating the
\texttt{SuspendedDepths} data structure with a work-stealing scheduler.
The basic idea is avoid contention with a ``copy-on-steal'' approach, where
each processor locally stores the portion of the data-structure it needs.


\paragraph{Implementing \texttt{SuspendedDepths}}
Throughout execution, each processor's current task is a leaf in the
task tree.
A key observation is that, for our order-maintenance algorithm, each processor
only needs access to the dag-depths of tasks along its root-to-leaf path in the
tree.
Therefore, we can implement \texttt{SuspendedDepths} in a distributed fashion,
using one array stored locally on each processor, where each array (of
length up to $F$, the maximum nesting depth) stores a root-to-leaf path of suspended
depths.

Each processor-local array operates like a stack.
Immediately before each fork, the current processor pushes the current
dag-depth onto its local array.
Immediately after each join, the current processor pops one element.
To guarantee that processors locally have the suspended depths of all
ancestor tasks, we copy a small amount of data on each steal:
when a processor steals a task that has $K$ ancestors, it copies $K$
suspended depths from the victim into its own local array.
This requires $O(F)$ work on each steal, which in practice is effectively
free, because $F$ is typically small, and the cost of steals is often
well-amortized with appropriate granularity control.

\section{Related Work}
\label{sec:related}
In 1974, Brent proved that a parallel program with total work $W$ and
span (depth) $S$ can be executed on $P$-processor, in
$T_P \leq W/P + S$ time.
This seminal result has underpinned decades of research on scheduling
algorithms for task parallelism.
This research have shown that this bound is not purely theoretical and
can in fact be realized efficiently in
practice~\cite{burtonsl81,halstead85,frigolera98,blumofele99,narlikarbl99,abp-multi-01,abb02,frs-schedule-04,acarchra13,acarchra16oracle},
usually by using a variant of work stealing.
Recent work has also made some progress extending this scheduling
theory to account for latency~\cite{ma-latency-2016} and competition
among tasks, by allowing tasks to declare
priorities~\cite{mwa-fairness-2019,singer+priority-2020}.

The advances on scheduling algorithms it turn lead to the development
of many task-parallel programming languages and systems in both
functional and procedural languages.

An important concern is all these task-parallel systems is race detection.
Because data races usually cause incorrect behavior, especially at
scale~\cite{adve-races-2010,boehm-miscompile-2011}, there has been
much work on detecting races in task parallel, such as fork-join
programs, have been
proposed~\cite{Mellor-Crummey91,FengLe97,ChengFeLe98,BenderFiGi04,Fineman05,RamanZhSa10,RamanZhSa12,uaf+race-2016,xsl-determinacy-2020}.
These algorithms all revolve around an ordering data structure,
that allows determining whether two instructions are
sequentially dependent or can be executed in parallel.
The basic idea is to check at each memory access whether the
instruction performing the access creates a race with the instructions
that accessed the same location in the past by using the ordering data
structure.
The theoretical and practical efficiency of these algorithms therefore
critically depend on the series-parallel order data structure.
Of the prior approaches, some do not guarantee constant work/time
bounds~\cite{Mellor-Crummey91,RamanZhSa10,RamanZhSa12,ljs-enforcement-2014}.
Others run only sequentially~\cite{FengLe97,ChengFeLe98} or limit
parallelism with non-constant overheads~\cite{BenderFiGi04}.
More recent work can guarantee constant overheads by carefully
integrating the ordering data structure and the scheduler but this
comes with significant complexity in managing
concurrency~\cite{uaf+race-2016}.
All of the above work considers nested parallelism with fork-join and
async-finish constructs, which result in similar dependency
structures.
More recent work considers race detection for futures and establishes
worst-case bounds, though the overheads are no longer
constant~\cite{xsl-determinacy-2020}.

In comparison to this related work, our DePa algorithm is perhaps most
similar to Mellor-Crummey's \emph{offset-span} vertex labeling
technique~\cite{Mellor-Crummey91}.
An offset-span labeling for a vertex $v$ requires $O(f_v)$ space, where
$f_v$ is the dynamic nesting depth of that vertex; similarly, the worst-case
time required for a query is $O(F)$ where $F$ is the maximum dynamic nesting
depth.
These bounds may at first appear essentially identical to those achieved by
DePa, if one ignores the factor $\omega$ (the word-size) in our bounds.
However, we specifically included this factor in the bounds as it contributes
to a significant savings in practice.
For any vertex $v$, the offset-span label of $v$ requires approximately
$2 f_v$ words; in contrast, the DePa label for the same vertex requires
only $1 + \lceil f_v/\omega \rceil$ words.
For typical word sizes (e.g. $\omega = 64$), this is between one and two orders
of magnitude improvement in space.
The cost of queries is likely also reduced, although a direct comparison
will be needed to determine the exact difference in overhead.

Race detection has also been studied extensively in the more general
concurrency setting, where programs countain a small number of
coarse-grained threads that may synchronize by using locks,
synchronization variables, etc.
Techniques from this second line of work do not scale to task-parallel
programs, because of their coarse-grained threads assumption
(e.g,.~\cite{RamanZhSa12}), and is less directly relevant to this
paper.

Early work on coarse-grained threads proposes the lock-set
algorithm~\cite{SavageBuNe97}, which can lead to false positives.
Subsequent work proposed precise techniques by using vector-clocks to
capture the happens-before relations between
threads \cite{FlanaganFr09}.
Followup work has proposed hybrid approaches that combine lock sets
and vector clocks, trading off efficiency and
precision~\cite{cc-hybrid-2003,yrc-2005-racetrack}.
Most approaches use dynamic, on-the-fly race detection though there
has also been some work on predicting data
races~\cite{ses+sound-2012,kmv-2017-dynamic}.

Dynamic race-detection techniques in coarse-grained multithreaded
programs are typically sensitivity to scheduling decisions: because
they track actual threads, the techniques detect races in the run
being checked and remain sensitive to scheduling decisions.
In contrast, prior race detection techniques for task parallel
programs are able to account for logical (potentially unrealized)
parallelism~\cite{Mellor-Crummey91,FengLe97,ChengFeLe98,BenderFiGi04,RamanZhSa10,RamanZhSa12,ljs-enforcement-2014,uaf+race-2016}.

\paragraph{Pedigrees and DPRNGs}
Any scheme which uniquely identifies individual ``vertices'' of a computation
graph while ignoring details of scheduling is known as
a \defn{pedigree} scheme~\cite{LeisersonScSu12}.
Pedigrees are potentially useful for a variety of applications.
One such application is a
\textbf{d}eterministic \textbf{p}arallel
\textbf{r}andom-\textbf{n}umber \textbf{g}enerator
(or \defn{DPRNG}, for short), where pedigrees
can be used as seeds for the generator.
For example, Leiserson et al.~\cite{LeisersonScSu12} describe a particular
pedigree scheme and use these pedigrees to implement a DPRNG with low overhead
in practice.
DPRNGs make it possible
to write parallel randomized algorithms with repeatable behavior:
because pedigrees are independent of scheduler decisions, this approach
is deterministic under parallel execution.
This kind of repeatable behavior is a form of
\emph{internal determinism}, which
provides a number of benefits, such as simplifying the process of debugging
and performance tuning~\cite{bfgs12-pbbs}.

Because DePa provides a unique labeling of vertices within a computation graph,
it could be used as an efficient pedigree scheme for a DPRNG.
In comparison to the pedigree scheme of Leiserson et al.
(\cite{LeisersonScSu12}), one potential advantage of DePa is in reducing the
size of pedigrees; specifically, DePa labels are approximately a factor
$\omega$ smaller, where $\omega$ is the word-size of the machine.

\section{Conclusion}
\label{sec:conc}
We present \defn{DePa}, a novel solution to the so-called \emph{SP maintenance}
problem which is key to the efficiency of techniques such as dynamic
race detection.
The gist of our approach is to represent a computation as a graph, where
each vertex is labeled with both is \emph{depth} in the graph as well as
the \emph{fork path} of its corresponding nested task.
This encoding is compact and enables fast precendence queries between
vertices.
DePa therefore provides a number of advantages in comparison to prior
techniques:
(1) it is provably efficient with low overhead in terms of both space and time,
(2) it is fully parallel,
and
(3) it is simple to implement.

\clearpage

\bibliography{local,../../../bibliography/main,../../../bibliography/new,../../../bibliography/leiserson,../../../bibliography/gc}



\end{document}